\documentclass{article}

\usepackage{graphicx}%
\usepackage{multirow}%
\usepackage{amsmath,amssymb,amsfonts}%
\usepackage{amsthm}%
\usepackage{mathrsfs}%
\usepackage{xcolor}%
\usepackage{textcomp}%
\usepackage{booktabs}%
\usepackage{algorithm}%
\usepackage{algorithmicx}%
\usepackage{algpseudocode}%
\usepackage{listings}%
\usepackage{multirow}
\usepackage{authblk}
\usepackage{natbib}
\usepackage[hidelinks]{hyperref}
\usepackage{appendix}
\usepackage{geometry}

\newtheorem{theorem}{Theorem}
\newtheorem{assumption}{Assumption}%
\newtheorem{proposition}{Proposition}

\newtheorem{definition}{Definition}%

\newcommand{\keywords}[1]{%
\par\noindent\textbf{Keywords: } #1
}

\begin{document}

\title{Post-Selection Inference for Multiverse Analysis in Mixed-Effects Models (PIMAX)}

\author[1]{Anna Vesely\thanks{anna.vesely2@unibo.it}}
\author[2]{Angela Andreella\thanks{angela.andreella@unive.it}}

\affil[1]{Department of Statistical Sciences, University of Bologna, Via delle Belle Arti 41, 40126 Bologna, Italy}
\affil[2]{Department of Economics, Ca' Foscari University of Venice, Cannaregio 873, 30100 Venezia, Italy}

\date{}


\maketitle

\abstract{Sign-flipping score tests provide robust inference in generalized linear models under variance misspecification and form the basis of two recent inferential frameworks: post-selection inference in multiverse analysis (PIMA) and the sign-flipping score-based two-stage summary-statistics approach (flip2sss). PIMA provides asymptotically valid inference across a multiverse of model specifications, whereas flip2sss extends sign-flipping score testing to longitudinal and hierarchical data through cluster-level summary statistics. In this paper, we combine these two approaches to develop PIMAX, a multiverse inferential framework for clustered observations. The resulting method extends post-selection inference to clustered-data settings, accommodating heteroscedasticity, unbalanced designs, and within-cluster dependence. Given a multiverse of candidate specifications, PIMAX provides a global $p$-value for testing whether any specification exhibits a non-zero effect (weak control of the family-wise error rate, FWER), lower confidence bounds on the number of true discoveries, and multiplicity-adjusted $p$-values for identifying the specific contributing specifications (strong FWER control). By avoiding inference based on a fully specified random-effects covariance structure, PIMAX solves a key source of type I error inflation due to random-effects misspecification while enabling inference across a multiverse of fixed-effects specifications.}




\section{Introduction}

High-dimensional and structured data settings often involve a large number of defensible analytical choices. Different selections of covariates, transformations, adjustment sets, or modeling assumptions may all be scientifically justified, yet lead to different substantive conclusions. Multiverse analysis \citep{steegen2016increasing} addresses this issue by systematically evaluating results across a collection of plausible specifications, thereby making explicit the role of researcher degrees of freedom in the analysis. This perspective is especially relevant in light of longstanding concerns about selective reporting and undisclosed analytical flexibility \citep{simmons2011false,gelman2014garden,nosek2014scientific}. However, transparency alone does not resolve the inferential problem induced by model multiplicity. Once multiple specifications have been considered, formal inference must account for the fact that conclusions are drawn after examining a set of reasonable alternatives. Multiverse analysis therefore naturally gives rise to a post-selection inference problem \citep{benjamini2020selective}, in which multiplicity must be controlled over the full collection of candidate models.

Post-selection inference in multiverse analysis (PIMA), proposed by \citet{girardi2024post}, addresses this problem in the framework of generalized linear models (GLMs). Building on sign-flipping score tests \citep{hemerik2020robust,de2025inference}, PIMA provides asymptotically valid inference across a multiverse of model specifications, allowing one to test whether the effect of interest is present in at least one plausible specification and to perform model-specific inference with appropriate multiplicity adjustment. Operationally, PIMA implements a resampling-based test within each specification using sign-flip transformations and combines the resulting standardized statistics across the multiverse via a suitable combining function. With proper handling of the transformations, this construction preserves dependence among model-specific statistics and yields a global test of the intersection null hypothesis that the effect is absent in all candidate specifications. Through closed testing theory \citep{marcus1976closed,goeman2011multiple} and its associated shortcut methods \citep{westfall1993resampling,vesely2023permutation}, PIMA provides multiplicity-adjusted $p$-values with strong family-wise error rate (FWER) control and simultaneous lower confidence bounds on the number of false null hypotheses, offering post-selection identification and quantification of the specifications that support a non-zero effect. PIMA thus represents a formal inferential counterpart to multiverse analysis. An alternative approach is specification curve analysis \citep{simonsohn2020specification}. However, the method only provides a global test for the presence of at least one non-zero effect across specifications, with weak FWER control. Furthermore, its theoretical justification is developed in the context of linear models and does not readily extend to GLMs \citep{girardi2024post}. For these reasons, we focus our attention on PIMA.

A major limitation of the original PIMA framework is that it assumes independent observations. In many applications, however, data are clustered, longitudinal, or multilevel, so that valid inference must account for within-cluster dependence. In these settings, model uncertainty is inherently more complex. Researchers are typically uncertain not only about the fixed-effects component of the model, but also about the form of the dependence structure. Within the generalized linear mixed model (GLMM) framework, inferential conclusions may depend on whether random slopes are included, which covariates are allowed to vary across clusters, and how the covariance structure of the random effects is specified \citep{breslow1993approximate,mcculloch2008generalized}. These choices are rarely unique in practice, and different defensible specifications may lead to different conclusions, as reflected in the methodological debate on maximal versus parsimonious random-effects structures \citep{barr2013random,bates2015parsimonious,matuschek2017balancing}. Moreover, likelihood-based inference in GLMMs may be sensitive to misspecification of the random-effects structure, while estimation itself may become unstable in complex models, sparse binary settings, or small samples, where convergence failures, singular fits, and boundary estimates are common \citep{breslow1993approximate,cnaan1997tutorial}. Both cases lead to inflated type I error, i.e., invalid inference on fixed effects. Thus, in the presence of dependent observations, a multiverse is shaped not only by alternative mean structures but also by competing assumptions on the underlying dependence mechanism. Any inferential procedure for this setting should therefore handle clustering and multiplicity simultaneously, while avoiding undue reliance on a single fully specified random-effects model.

A partial solution is provided by the sign-flipping score-based two-stage summary-statistics approach (flip2sss) of \cite{andreella2025robust}, which extends resampling-based score testing to dependent data. The procedure follows a two-stage strategy: cluster-specific summary measures are first obtained by fitting separate GLMs within clusters and are then analyzed jointly in a second-stage score test based on sign flipping. This construction reduces the dependence problem to the cluster level and avoids explicit parametric specification of the random-effects distribution or covariance structure, while accommodating heteroscedasticity and unbalanced designs. However, flip2sss is not a post-selection method: it yields valid inference within a given specification, but does not address the multiplicity induced by a multiverse of plausible models.

In this paper, we combine these two lines of work to develop PIMAX, a multiverse inferential framework for clustered data. By embedding the two-stage summary-statistics construction of flip2sss within the post-selection machinery of PIMA, PIMAX provides inference across a multiverse of clustered-data specifications while accounting simultaneously for within-cluster dependence and model multiplicity. The procedure accommodates heteroscedasticity and unbalanced designs, and avoids requiring a fully specified random-effects distribution or covariance structure, making it particularly suited to settings in which several fixed-effects specifications are scientifically defensible. It is implemented in the \texttt{R} package \texttt{pima} \citep{pima}.

The remainder of the paper is organized as follows. Section \ref{setting} introduces the notation and background. Section \ref{flip2sss} revisits the flip2sss procedure, extends its theoretical characterization, and recasts it in terms of cluster scores. Then Section \ref{pimax} builds on this foundation to introduce and formally define the PIMAX framework. Section \ref{simulations} investigates its finite-sample validity and power through simulations and compares its performance with competing methods, while Section \ref{application} presents an application to the survey of health, ageing, and retirement in Europe (SHARE) data \citep{borsch2013data}. Finally, Section \ref{discussion} concludes with a discussion.

\section{Setting}\label{setting}

We consider clustered data, such as longitudinal data, where responses within the same cluster may be dependent. Let $j \in \{1,\dots,J\}$ index the clusters, and let $i \in \{1,\dots,n_j\}$ index the units within cluster $j$, with total sample size $n=\sum_{j=1}^J n_j$. For unit $i$ in cluster $j$, let $Y_{ij}$ denote the response variable.

To fix notation, we consider a GLMM representation in which, conditionally on cluster-specific random effects, $Y_{ij}$ follows a distribution in the exponential family and its conditional mean $\mu_{ij}$ satisfies

\begin{equation}\label{eq:model}
g(\mu_{ij}) = \xi + x_{ij}\beta + \boldsymbol{z}_{ij}^{\top}\boldsymbol{\gamma} + \boldsymbol{w}_j^{\top}\boldsymbol{\delta} + U_j + G_j x_{ij} + \boldsymbol{D}_{j}^{\top}\boldsymbol{z}_{ij},
\end{equation}
where $g(\cdot)$ is a known link function. Here, $x_{ij}\in \mathbb{R}$ is a covariate of interest and $\beta\in \mathbb{R}$ is its fixed effect. The covariate may vary within clusters ($x_{ij}\neq x_{i'j}$ for some $i,i',j$) or be cluster-constant ($x_{ij} = x_j$); throughout, we use the notation $x_{ij}$ to cover both cases. Furthermore, $\xi \in \mathbb{R}$ represents the intercept, $\boldsymbol{z}_{ij}\in\mathbb{R}^{p}$ and $\boldsymbol{w}_j\in\mathbb{R}^{q}$ denote within- and between-cluster nuisance covariates, respectively, and $\boldsymbol{\gamma}\in\mathbb{R}^{p}$ and $\boldsymbol{\delta}\in\mathbb{R}^{q}$ collect the associated fixed effects. The cluster-specific random variables $U_j, G_j\in\mathbb{R}$ and $\boldsymbol{D}_j \in \mathbb{R}^{p}$ represent a random intercept, a random deviation in the effect of $x_{ij}$, and analogous deviations in the effects of the within-cluster nuisance covariates, respectively. This formulation is representative of a standard GLMM setup \citep{mcculloch2008generalized}.

We remark here that the model in \eqref{eq:model} is introduced primarily to fix notation and formalize the inferential setting. It should not be interpreted as a structural assumption required by the proposed method. In particular, the random-effects specification is only one possible working representation of the within-cluster dependence. The inferential framework developed below is compatible with submodels of \eqref{eq:model}, as well as with richer dependence structures obtained by extending the random component. This is because the proposed procedure does not rely on an explicit parametric specification of the random-effects distribution or covariance structure.

The inferential target is the effect of the covariate of interest on the response. Specifically, we consider the null hypothesis
\begin{equation}\label{eq:hyp_uni}
H_0:\beta=0
\qquad\text{against}\qquad
H_1:\beta\neq 0.
\end{equation}
The interpretation of $\beta$ depends on the level at which $x_{ij}$ varies. If it varies across units within clusters, $\beta$ represents an effect identified by within-cluster variation. In contrast, if $x_{ij}=x_j$ is constant within clusters, $\beta$ represents a between-cluster effect. The proposed framework extends directly to the more general null hypothesis $H_0:\beta=\beta_0$ and to one-sided alternatives, but we restrict our attention to \eqref{eq:hyp_uni} for simplicity of exposition. 

Under a conventional GLMM analysis, inference on $\beta$ is typically derived from asymptotic likelihood theory, e.g., via Wald, likelihood-ratio, or score statistics \citep{breslow1993approximate,mcculloch2008generalized}. The validity of these procedures depends critically on the correct specification of the random-effects structure \citep{heagerty2001misspecified,andreella2025robust,andreella2025blockwise,andreella2026multivariate}. In particular, omission of relevant random slopes, misspecification of the random-effects covariance matrix, or adoption of an overly restrictive dependence structure may yield biased variance estimates and, consequently, incorrect null approximations for the resulting test statistics, leading to anticonservative tests and inflated type I error. Thus, even when the fixed-effects component is correctly specified, inference on $\beta$ may be unreliable if the random part of the model is misspecified. This is methodologically consequential, since the random-effects structure is generally not identified by the data alone and must instead be selected among competing plausible parameterizations \citep{barr2013random,bates2015parsimonious,matuschek2017balancing}.

A second issue with GLMM-based testing procedures is computational. Because the GLMM likelihood generally involves integration over the random effects, estimation typically relies on numerical approximation \citep{breslow1993approximate,mcculloch2008generalized}. In finite samples, especially with binary or sparse responses, highly correlated covariates, or complex random-effects structures, this may produce non-convergence, singular fits, boundary estimates of variance components, nearly unidentified covariance parameters, and unstable standard error estimates. As a consequence, the finite-sample behavior of the resulting test statistics may deviate substantially from their nominal asymptotic approximations, with possible inflation of the type I error rate \citep{andreella2025robust,andreella2025blockwise}.

These difficulties become especially relevant when inference is carried out across a collection of plausible fixed-effects specifications. In many applications, there is uncertainty about which nuisance covariates should be included, how continuous covariates should be represented, or whether a covariate should enter the model as a within- or between-cluster effect. Accordingly, let
\begin{equation}\label{eq:multiverse}
\mathcal{M}=\{m_1,\dots,m_K\}
\end{equation}
denote a collection of plausible model specifications that may differ in the definition of the response $Y_{kij}$ (e.g., from outlier and/or leverage point removal), the predictors $x_{kij}$, $\boldsymbol{z}_{kij}$ and $\boldsymbol{w}_{kj}$ (combination and transformation), the link function $g_k$, and the estimation procedure used to fit the model (maximum likelihood estimation, Firth bias-reduced estimation \citep{firth1993bias}, penalized likelihood methods, etc.).

In this setting, the inferential problem is to perform valid inference across $\mathcal{M}$ while simultaneously accounting for within-cluster dependence and the multiplicity induced by considering multiple candidate specifications. In particular, let $\beta_k$ denote the effect of the considered covariate under specification $m_k$. Then interest lies in making confidence statements on the model-specific null hypotheses
\begin{equation}\label{eq:hyp_uni_k}
H_{0k}:\beta_k=0
\qquad\text{against}\qquad
H_{1k}:\beta_k\neq 0,\qquad k=1,\ldots,K
\end{equation}
at a pre-fixed confidence level $1-\alpha$, with $\alpha\in (0,1)$.

To address this challenge, we propose PIMAX, a multiverse inferential framework for clustered data. As shown in the next sections, PIMAX enables three types of confidence statements on the family of hypotheses $\{H_{0k}\}_{k=1}^K$: (i) a global $p$-value with weak FWER control (“Is there at least one specification with a non-zero effect?”), (ii) lower confidence bounds on the number or proportion of true discoveries (“How many?”), and (iii) multiplicity-adjusted $p$-values with strong FWER control (“Which ones?”). The adjusted $p$-values are a key feature of PIMAX, as they allow researchers to select and report a preferred specification post hoc, after seeing the data and performing the analysis, while maintaining valid type I error control.

\section{The flip2sss approach}\label{flip2sss}
The main component underlying PIMAX is the flip2sss method developed by \citet{andreella2025robust}. Given a single specification \eqref{eq:model} and the corresponding system of hypotheses \eqref{eq:hyp_uni}, flip2sss provides valid resampling-based inference on $H_0$ by reducing within-cluster dependence using cluster-level summaries, from which score statistics are then constructed, as defined below. We consider two cases: a within-cluster scenario, where the covariate of interest $x_{ij}$ varies within clusters, and a between-cluster scenario, where $x_{ij} = x_j$ is cluster-constant.

\subsection{Two-stage summaries}

\begin{definition}[First-stage model]
\label{def:first_model}
For each cluster $j=1,\dots,J$, define the first-stage GLM
\begin{equation*}
g(\mu_{ij})
=
\xi_{j}
+
\boldsymbol{r}_{ij}^{\top}\boldsymbol{\lambda}_{j},\qquad i=1,\ldots,n_j,
\end{equation*}
where $\boldsymbol{r}_{ij}\in\mathbb{R}^{d}$ contains the within-cluster covariates:
\begin{equation*}
\boldsymbol{r}_{ij} =
\begin{cases}
\left[
x_{ij},\;\boldsymbol{z}_{ij}^{\top}
\right]^\top & \text{in the within-cluster case}\\
\boldsymbol{z}_{ij} & \text{in the between-cluster case.}
\end{cases}
\end{equation*}
\end{definition}

\begin{definition}[First-stage summary statistic]
\label{def:summary_statistics}
For each cluster $j=1,\dots,J$, let
\begin{equation*}
\hat{\boldsymbol{\theta}}_{j}
=
\left[
\hat{\xi}_{j},\; \hat{\boldsymbol{\lambda}}_{j}^{\top}
\right]^{\top}
\in \mathbb{R}^{d+1}
\end{equation*}
denote the vector of cluster-specific estimators obtained from the first-stage GLM, and let
\begin{equation*}
t_{j}=\boldsymbol{a}^\top\hat{\boldsymbol{\theta}}_{j}\in\mathbb{R}
\end{equation*}
be a scalar first-stage summary statistic, where $\boldsymbol{a}\in\mathbb{R}^{d+1}$ is a contrast. In particular, let
\begin{equation*}
t_j =
\begin{cases}
\hat{\beta}_{j} & \text{in the within-cluster case}\\
\hat{\xi}_{j} & \text{in the between-cluster case,}
\end{cases}
\end{equation*}
where $\hat{\beta}_{j}$ denotes the estimate of the cluster-specific coefficient of $x_{ij}$.
\end{definition}

\begin{definition}[Second-stage working model]
\label{def:second_stage}
Consider a second-stage linear model in which $t_j$ serves as the response and the between-cluster covariates enter as predictors:
\begin{equation}
\label{eq:within_between}
t_j =\tilde{\beta} q_j + \boldsymbol{u}_j^\top \boldsymbol{\varphi} + \varepsilon_j,\qquad\varepsilon_j\sim (0,v_j)\text{ independent,}
\end{equation}
where $q_j\in\mathbb{R}$ is the target variable, constructed to capture the dependence of the original model \eqref{def:first_model} on $x_{ij}$, and $\boldsymbol{u}_j\in\mathbb{R}^{m}$ collects the nuisance covariates. Specifically,
\begin{equation*}
q_j;\; \boldsymbol{u}_j =
\begin{cases}
1;\; \boldsymbol{w}_j & \text{in the within-cluster case}\\
x_j;\; \left[1,\boldsymbol{w}_j^\top\right]^\top & \text{in the between-cluster case.}
\end{cases}
\end{equation*}
\end{definition}

To facilitate readability, Table \ref{tab:within_between_summary} summarizes the key quantities introduced in Definitions \ref{def:first_model}--\ref{def:second_stage}.

\begin{table}
\centering
\caption{Summary of quantities used in the flip2sss method in the within- and between-cluster cases.}
\label{tab:within_between_summary}
\begin{tabular}{llll}
\toprule
Stage & Quantity & Within & Between \\
\midrule

\multirow{4}{*}{First} & Covariates $\boldsymbol{r}_{ij}$ 
& $[x_{ij}, \boldsymbol{z}_{ij}^\top]^\top$ 
& $\boldsymbol{z}_{ij}$ \\

& Coefficients $\boldsymbol{\theta}_j=[\xi_j,\boldsymbol{\lambda}_j^\top]^\top$
& $[\xi_j,\beta_j,\boldsymbol{\gamma}_j^\top]^\top$ 
& $[\xi_j,\boldsymbol{\gamma}_j^\top]^\top$ \\

& Model $g(\mu_{ij}) = \xi_j + \boldsymbol{r}_{ij}^\top\boldsymbol{\lambda}_j$ 
& $g(\mu_{ij}) = \xi_j + x_{ij}\beta_j + \boldsymbol{z}_{ij}^\top\boldsymbol{\gamma}_j$ 
& $g(\mu_{ij}) = \xi_j + \boldsymbol{z}_{ij}^\top\boldsymbol{\gamma}_j$ \\

& Summary statistic $t_j$ 
& $\hat{\beta}_j$ 
& $\hat{\xi}_j$ \\

\midrule

\multirow{3}{*}{Second}

& Target variable $q_j$ & $1$ & $x_j$ \\

& Nuisance covariates $\boldsymbol{u}_j$ & $\boldsymbol{w}_j$ & $[1,\boldsymbol{w}_j^\top]^\top$ \\

& Model $t_j = \tilde{\beta} q_j + \boldsymbol{u}_j^\top \boldsymbol{\varphi} + \varepsilon_j$ 
& $t_j = \tilde{\beta} + \boldsymbol{w}_j^\top \boldsymbol{\delta} + \varepsilon_j$ 
& $t_j = \tilde{\xi} + x_j \tilde{\beta} + \boldsymbol{w}_j^\top \boldsymbol{\delta} + \varepsilon_j$ \\
\bottomrule
\end{tabular}
\end{table}

The second-stage model \eqref{eq:within_between} is treated as a working model for the cluster-level summaries. Throughout, we assume the following.

\begin{assumption}\label{ass2}
The conditional mean of the summary statistic $t_{j}$ is correctly specified by the corresponding second-stage model \eqref{eq:within_between}.
\end{assumption}

Assumption \ref{ass2} concerns only the mean structure of the cluster-level summaries. In particular, it does not require the correct specification of the random-effects distribution or the within-cluster covariance structure. Furthermore, observe that the error terms in \eqref{eq:within_between} are only assumed to be independent with zero mean, allowing for arbitrary heteroscedasticity and non-identical distributions.

Inference on $\beta$ in the original specification \eqref{eq:hyp_uni} is operationally carried out by studying $\tilde{\beta}$ in the second-stage working model \eqref{eq:within_between}:
\begin{equation}
\label{eq:hyp_uni_tilde}
\tilde{H}_0:\tilde{\beta}=0
\qquad\text{against}\qquad
\tilde{H}_1:\tilde{\beta}\neq 0.
\end{equation}

\begin{proposition}
\label{prop:null_equivalence}
Under Assumption \ref{ass2} and for the first-stage summaries in
Definition \ref{def:summary_statistics}, $H_0$ and $\tilde{H}_0$ are equivalent, in the sense that they induce the same restriction on the distribution of $\{t_j\}_{j=1}^J$.
\end{proposition}

Although the cluster-level summaries in Definition \ref{def:summary_statistics} are obtained here from cluster-specific GLMs, the flip2sss construction is not tied to this particular choice \citep{senn1990analysis}. Other first-stage reductions may be used, provided that they yield well-defined second-stage score contributions satisfying the validity conditions stated in Section \ref{subsec:score_test}. When several such reductions are scientifically defensible, the choice of first-stage summary may itself be included among the analytical decisions defining the multiverse. This flexibility is consistent with results in related settings showing that, under suitable conditions, procedures based on summary statistics need not be less efficient than analyses based on individual-level data \citep{lin2010relative}.

The validity and finite-sample performance of flip2sss nevertheless depend on the quality of the first-stage summaries. In small samples or sparse designs, cluster-level estimators may be biased or unstable, thereby affecting the centering and variability of the second-stage score statistic. Exact unbiasedness is not required, but the cumulative bias of the summaries entering the second-stage model must be asymptotically negligible, as formalized in the following.

\begin{proposition}\label{prop:bias}
Consider any second-stage statistic of the form
\begin{equation*}
S = J^{-1/2}\sum_{j=1}^J f_{j}(t_{j}),
\end{equation*}
where $f_{j}:\mathbb{R}\to\mathbb{R}$ is linear for each $j$ and the sequence $\{f_{j}\}_{j=1}^J$ is uniformly bounded. Let $t_j^\ast$ and $S^\ast$ be the oracle versions of $t_j$ and $S$, respectively, obtained by plugging-in the true value of $\boldsymbol{\theta}_{j}$ under $H_0$. If, as $J\to\infty$,
\begin{equation*}
J^{-1/2}\sum_{j=1}^J
\left|
\mathbb{E}(t_{j})-t_{j}^\ast
\right|
\to 0,
\end{equation*}
then
\begin{equation*}
\mathbb{E}(S)
-
\mathbb{E}(S^\ast)
\to 0.
\end{equation*}
\end{proposition}

Proposition \ref{prop:bias} shows that exact unbiasedness of the first-stage summaries $t_j$ is not required. What is needed is that the cumulative bias of the scalar summaries entering the second-stage model be negligible at the $J^{-1/2}$ scale of the second-stage statistic $S$. In particular, the condition in Proposition \ref{prop:bias} is implied by the stronger requirement $\max_{1\le j\le J}|\mathbb{E}(t_{j})-t_{j}^{\star}|=o(J^{-1/2})$. Moreover, if $|\mathbb{E}(t_{j})-t_{j}^{\star}|=O(n_j^{-1})$ uniformly in $j$, then a sufficient condition is $J^{-1/2}\sum_{j=1}^J n_j^{-1}\to 0$. If the cluster sizes are fixed or uniformly bounded, the rate $O(n_j^{-1})$ is not sufficient to guarantee this condition, since the within-cluster bias does not vanish as $J$ increases. In such settings, the first-stage summaries must therefore be chosen or corrected so that their cumulative bias remains negligible at the $J^{-1/2}$ scale.

\subsection{Score-based test}\label{subsec:score_test}
A test for $\tilde{H}_0$, and thus for $H_0$, is constructed via score-based sign-flipping \citep{hemerik2020robust,de2025inference} at the cluster level. The required assumptions and definitions are introduced below.

\begin{assumption}\label{ass1}
The clusters $j=1,\dots, J$ are mutually independent.
\end{assumption}


\begin{definition}[Cluster-wise score]
\label{def:test_uni}
In the second-stage model \eqref{eq:within_between}, let $\hat{\zeta}_{j}=\boldsymbol{u}_j^\top\hat{\boldsymbol{\varphi}}$ denote the fitted mean of $t_j$ under $H_0$, and recall that $v_j$ is the variance of $t_j$. The cluster-wise effective score contribution of the $j$th cluster is
\begin{equation}\label{eq:contribution}
c_{j}= \left[ q_{j} - \left(
\sum_{\ell=1}^J \frac{q_{\ell}\boldsymbol{u}_{\ell}^{\top}}{v_{\ell}}
\right)
\left( \sum_{\ell=1}^J \frac{\boldsymbol{u}_{\ell}\boldsymbol{u}_{\ell}^{\top}}{v_{\ell}}
\right)^{-1} \boldsymbol{u}_{j} \right] v_{j}^{-1}(t_{j}-\hat{\zeta}_{j}).
\end{equation}
The corresponding effective score statistic \citep{hemerik2020robust} is
\begin{equation*}
S = J^{-1/2}\sum_{j=1}^J c_{j}.
\end{equation*}
The standardized score statistic \citep{de2025inference}, introduced to improve small-sample performance, is obtained by dividing the effective score by its standard deviation:
\begin{equation*}
T = \frac{S}{\sqrt{\mathrm{Var}(S)}},
\end{equation*}
where
\begin{equation*}
\mathrm{Var}(S) = J^{-1} \sum_{j=1}^J
\tilde q_{j}^{\,2}v_{j}^{-1},\qquad \tilde q_{j} = q_{j} - \left(
\sum_{\ell=1}^J \frac{q_{\ell}\boldsymbol{u}_{\ell}^{\top}}{v_{\ell}}
\right) \left( \sum_{\ell=1}^J \frac{\boldsymbol{u}_{\ell}\boldsymbol{u}_{\ell}^{\top}}{v_{\ell}}
\right)^{-1} \boldsymbol{u}_{j}.
\end{equation*}
In practice, when the variance $v_j$ is unknown, it can be replaced by a $\sqrt{J}$-consistent estimate $\hat{v}_j$ \citep{de2025inference}. 
\end{definition}

\begin{assumption}\label{ass3}
Denote by $c_{j}^{\ast}$ the oracle score contribution, obtained by replacing all estimated second-stage quantities in \eqref{eq:contribution} by their population counterparts under $H_0$.
As $J\to\infty$,
\begin{equation*}
\frac{1}{J}\sum_{j=1}^J \mathrm{Var}(c_{j}^{\ast}) \to c
\end{equation*}
for some constant $c>0$. Moreover, for every $\varepsilon>0$,
\begin{equation*}
\frac{1}{J}\sum_{j=1}^J
\mathbb{E}\left[
(c_{j}^{\ast})^2
\mathbf{1}\left\{
\frac{|c_{j}^{\ast}|}{\sqrt{J}}>\varepsilon
\right\}
\right]
\to 0.
\end{equation*}
\end{assumption}



Assumption \ref{ass3} imposes regularity conditions on the cluster-wise score contributions. It requires the effective score $S$ to have a nondegenerate asymptotic variance and rules out settings in which the limiting behavior of the test is driven by only a few clusters. 


A resampling-based test is then constructed by approximating the null distribution of the standardized score statistic $T$ by randomly flipping the signs of the cluster-wise score contributions $c_j$.

\begin{definition}[Standardized score test]
\label{def:test}
Fix a number of transformations $B\in\mathbb{N}$. For any $j=1,\ldots,J$, let $f_{j1}=1$ and, for $b=2,\dots,B$, let $f_{jb}\in\{-1,+1\}$ be independent Rademacher random variables. The $b$th sign-flipped version of the effective score is
\begin{equation}
\label{eq:score_null}
S_{b}=J^{-1/2}\sum_{j=1}^J f_{jb}c_{j}.
\end{equation}
Its standardized version is
\begin{equation}
\label{eq:score_std_null}
T_{b}
=
\frac{S_{b}}
{\sqrt{\mathrm{Var}(S_{b}\mid f_{1b},\dots,f_{Jb})}},
\end{equation}
with
\begin{equation}
\label{eq:var_score_null}
\mathrm{Var}(S_{b}\mid f_{1b},\dots,f_{Jb})
=
J^{-1}
\sum_{j=1}^J\sum_{\ell=1}^J
\tilde q_{j}v_{j}^{-1/2}
\,\hat a_{j\ell}^{(b)}\,
v_{\ell}^{-1/2}\tilde q_{\ell}
\end{equation}
and
\begin{equation*}
\hat a_{j\ell}^{(b)}
=
\sum_{r=1}^J
\left(\mathbf{1}{\{j=r\}}-\hat h_{jr}\right)\,f_{rb}
\sum_{s=1}^J
\left(\mathbf{1}\{r=s\}-\hat h_{rs}\right)\,f_{sb}\,
\left(\mathbf{1}\{s=\ell\}-\hat h_{s\ell}\right).
\end{equation*}
where $\hat h_{jr}$ denotes the $(j,r)$th entry of the weighted hat matrix.

Then $H_0$ is rejected at level $\alpha$ if
\begin{equation*}
T_1 < T_{(\lceil \alpha B/2 \rceil)}\quad\text{or}\quad T_1 > T_{(\lceil (1-\alpha/2)B \rceil)}
\end{equation*}
where $T_1$ is the observed statistic, and
$T_{(1)} \leq \cdots \leq T_{(B)}$ are the ordered values of the $B$ sign-flipped statistics. A corresponding two-sided $p$-value is given by
\begin{equation*}
p
=
\frac{1}{B}
\sum_{b=1}^B
\mathbf{1}
\left\{
\left|T_{b}\right|
\geq
\left|T_{1}\right|
\right\}.
\end{equation*}
\end{definition}

One-sided $p$-values and rejection rules are obtained analogously. In the special case with no nuisance covariates in the second-stage model, the test simplifies significantly, as the score contributions \eqref{eq:contribution} reduce to $c_j=q_j t_j/v_j$ while the variance \eqref{eq:var_score_null} becomes
\begin{equation*}
\mathrm{Var}(S_{b}\mid f_{1b},\dots,f_{Jb})
=
J^{-1}
\sum_{j=1}^J
\tilde q_{j}^{\,2}v_{j}^{-1}
\end{equation*}
and no longer depends on the sign-flip transformation.

\begin{theorem}
\label{thm:local_validity}
Under Assumptions~\ref{ass2}--\ref{ass3}, the sign-flipping test of Definition \ref{def:test} is a valid $\alpha$-level test for $H_0$, asymptotically as $J\to\infty$.
\end{theorem}

This construction combines the asymptotic validity of sign-flipping score tests with a two-stage summary-statistics representation that shifts dependence from the observation level to the cluster level without requiring specification of the random-effects distribution or covariance structure. The validity of the test is asymptotic in the number of clusters $J$, while it remains exact for any fixed number of sign-flip transformations $B$. Large values of $B$ tend to provide higher power; to ensure a non-degenerate test with positive power, it is required that $B\geq 1/\alpha$ \citep{hemerik2020robust}. Importantly, Proposition \ref{prop:bias} ensures robustness to small first-stage biases, provided that their cumulative effect is asymptotically negligible, and no distributional assumptions are imposed on the second-stage residuals beyond cluster-level independence and moment conditions.

\section{Inference in a multiverse of mixed-effects models}\label{pimax}
Consider the multiverse of candidate specifications \eqref{eq:multiverse} and the model-specific null hypotheses \eqref{eq:hyp_uni_k}. The corresponding global null hypothesis is the intersection null
\begin{equation*}
H_{\mathcal{M}}=\bigcap_{k=1}^K H_{0k},
\end{equation*}
that is, the hypothesis that the effect of interest is zero in all candidate specifications. A global test for $H_{\mathcal{M}}$ is constructed by combining information from the $K$ candidate specifications via model-specific standardized scores, following a procedure similar to PIMA \citep{girardi2024post}.

\begin{definition}[Global test]
\label{def:global_combined_statistic}
Draw $B$ random sign-flipping transformations as in Definition \ref{def:test}: for each cluster $j=1,\ldots,J$, $f_{j1}=1$, while $f_{jb}\in\{-1,+1\}$ are independent Rademacher random variables for $b=2,\ldots,B$. For each specification $m_k$ and each transformation $b$, let $T_{kb}$ be the corresponding standardized score \eqref{eq:score_std_null}. Then define the $b$th combined test statistic as
\begin{equation*}
G_b
=
\psi\left(|T_{1b}|,\ldots,|T_{Kb}|\right),
\end{equation*}
where $\psi:\mathbb{R}^K\to\mathbb{R}$ is a combining function that is non-decreasing in each argument, such as the mean and the maximum. The global null $H_{\mathcal{M}}$ is rejected at level $\alpha$ if
\begin{equation*}
G_1 > G_{(\lceil (1-\alpha)B \rceil)}
\end{equation*}
where $G_{(1)} \leq \cdots \leq G_{(B)}$. The corresponding global $p$-value is
\begin{equation*}
p_{\mathcal{M}}
=
\frac{1}{B}\sum_{b=1}^B \mathbf{1}\{G_b\ge G_1\}.
\end{equation*}
\end{definition}

The use of common sign-flip transformations across specifications preserves the dependence among model-specific score statistics and provides a resampling approximation to their joint null distribution. The combining function $\psi$ must be specified a priori, and its choice may impact the power properties of the global test. It may be defined directly on the standardized score statistics $T_{1b},\ldots, T_{Kb}$ or, alternatively, obtained by transforming them into $p$-values $p_{1b},\ldots,p_{Kb}$, either via parametric inversion or rank-based methods, and then applying classical $p$-value combination rules such as Fisher's method \citep{pesarin2001multivariate,girardi2024post}.

\begin{theorem}
\label{thm:pimax_validity}
Suppose that Assumptions~\ref{ass2}--\ref{ass3} hold for every specification $m_k\in\mathcal{M}$. Then the global test based on Definition \ref{def:global_combined_statistic} is a valid $\alpha$-level test for $H_{\mathcal{M}}$, asymptotically as $J\to\infty$.
\end{theorem}

Theorem \ref{thm:pimax_validity} extends local validity of inference in flip2sss to the multiverse level. Rejection of the global null hypothesis indicates evidence that the effect of interest is non-null in at least one of the considered specifications, with weak FWER control.

The same approach extends to arbitrary subsets of specifications. For any non-empty index set $I\subseteq\{1,\dots,K\}$, define the sub-collection $\mathcal{I}=\{m_k\in\mathcal{M}\,:\,k\in I\}$ and the corresponding intersection null hypothesis
\begin{equation*}
H_{\mathcal{I}}
=
\bigcap_{k\in I} H_{0k}.
\end{equation*}
A global test for $H_{\mathcal{I}}$ is obtained analogously to Definition \ref{def:global_combined_statistic}, by computing a subset-specific combined statistic
\begin{equation*}
G^{I}_b
=
\psi_I\left((|T_{kb}|)_{k\in I}\right),
\end{equation*}
where $\psi_I:\mathbb{R}^{|I|}\to\mathbb{R}$, and comparing the observed value $G^I_1$ with the sign-flip distribution. Under the same conditions as in Theorem \ref{thm:pimax_validity}, this yields an asymptotically valid test of $H_{\mathcal{I}}$.

These subset tests provide the building blocks for closed testing procedures \citep{marcus1976closed}. As the tests are asymptotically valid, closed testing yields asymptotically valid multiplicity-adjusted $p$-values for each specification, with strong control of the FWER over $\mathcal{M}$. In this way, it allows researchers to identify the specifications with a non-zero effect. Closed testing enjoys an optimality property in the sense that any multiple testing method controlling the FWER or related error measures is either equivalent to a closed testing procedure or can be improved by one \citep{goeman2021only}. In practice, full closed testing can be computationally demanding and may become infeasible when the number $K$ of specifications is large. In such cases, one can adopt the combination function $\psi=\max$ and rely on the corresponding maxT algorithm, which represents a computationally efficient shortcut \citep{westfall1993resampling}.

Furthermore, closed testing and its shortcuts provide a multiverse-level summary of the evidence, as they can be used to derive simultaneous lower confidence bounds on the number of false null hypotheses, either in $\mathcal{M}$ or within subsets $\mathcal{I}\subseteq\mathcal{M}$ \citep{goeman2011multiple,blanchard2020post,blain2022notip,andreella2023permutation,vesely2023permutation}. Simultaneity ensures validity even when the subset is selected post hoc. This allows one to quantify how many candidate specifications support the presence of an effect. Although this approach may appear less informative than the maxT procedure, alternative choices of $\psi$ may be more sensitive to different signal structures and gain higher power in certain settings. Guidance on selecting the combining function based on the expected type of signal can be found in \citet{vesely2023permutation}.

Finally, the same approach extends directly to multiple parameters of interest. As argued by \citet{girardi2024post}, global test statistics can be defined for all individual parameters of interest by applying the same random sign-flipping transformations. This allows one to test an overall global null hypothesis of no effect in any model. Moreover, the resulting tests can be embedded within a closed testing framework to obtain multiplicity-adjusted $p$-values for each parameter and each model, as well as simultaneous lower confidence bounds on the number of true effects, both overall and stratified by model or by coefficient.

To summarize, PIMAX combines the cluster-level robustness of the two-stage sign-flipping approach of flip2sss with the multiverse resampling and post-selection guarantees of PIMA. The resulting procedure extends multiverse inference to clustered-data settings without requiring a fully specified random-effects distribution or covariance structure. We next examine the empirical behavior of PIMAX and compare it with existing alternatives.













\section{Simulations}\label{simulations}
The performance of PIMAX is evaluated in terms of type I error control and power and compared with GLMM-based inference \citep{bolker2009generalized}. Since the latter yields model-specific inference, multiplicity within the multiverse is handled by Holm-Bonferroni adjustment \citep{holm1979simple}. We examine two data-generating scenarios, corresponding to the within- and between-cluster settings considered in Definitions \ref{def:first_model}--\ref{def:second_stage}. 

For each design condition, the number of clusters is set to $J\in\{20,30,40\}$ and the common cluster size to $n_j\in\{10,20\}$. The response $Y_{ij}$ follows a Bernoulli distribution with success probability $\pi_{ij}\in(0,1)$ and canonical parameter $\eta_{ij}=\text{logit}(\pi_{ij})$ generated according to model \eqref{eq:model}, with nuisance covariates $z_{ij},w_j\in\mathbb{R}$ ($q=p=1$) and nuisance fixed effects $\xi=\gamma=\delta=0.5$. Cluster-specific random effects are generated from a trivariate Gaussian distribution as
$(U_j,G_j,D_j)^\top \sim \mathcal{N}(\mathbf{0},\Sigma)$ where $\Sigma = A R A$, with
$A=\mathrm{diag}(1.0,0.7,0.4)$ and $R$ an equicorrelation matrix with off-diagonal entries equal to $0.3$.

To mimic a multiverse of plausible and highly related models, we follow the latent-variable construction used in \citet{simonsohn2020specification} and subsequently in \citet{girardi2024post}. In each setting, the response is simulated from an underlying latent variable of interest $\ell_{ij}$, together with the nuisance covariates $z_{ij}$ and $w_j$, each correlated at level $0.60$ with $\ell_{ij}$. The latent variable is not directly observed; inference is instead based on proxy covariates derived from it. A multiverse of $K=100$ specifications is constructed, each defined by an alternative proxy $x_{kij}$ having correlation $0.85$ with $\ell_{ij}$. Operationally, the data are simulated in a sequential way, starting from $w_j\sim\mathcal{N}(0,1)$; then $\ell_{ij}$ is generated conditionally on $w_j$, and finally $z_{ij}$ and $x_{kij}$ as noisy transformations of $\ell_{ij}$. Table \ref{tab:sim_summary} summarizes this mechanism, presenting the variables in the order in which they are computed.

In the within-cluster scenario, $\ell_{ij}$ and $x_{kij}$ vary at the observation level, and the linear predictor is
\begin{equation*}
\eta_{ij} = \xi+\gamma z_{ij}+\delta w_j+\beta \ell_{ij}+U_j+G_j \ell_{ij}+D_j z_{ij}.
\end{equation*}
In contrast, in the between-cluster scenario, $\ell_j$ and $x_{kj}$ vary only at the cluster level, and
\begin{equation*}
\eta_{ij} = \xi+\gamma z_{ij}+\delta w_j+\beta \ell_{j}+U_j+D_j z_{ij}.
\end{equation*}
Each design condition is replicated $1000$ times.

\begin{table}
\centering
\caption{Simulations: data-generating mechanism. Here, $f_\rho(t)=\rho t + \sqrt{1-\rho^2}\,\varepsilon$, with $\varepsilon\sim\mathcal{N}(0,1)$. The errors $\varepsilon$ vary within- or between-cluster, according to the variable that is generated via $f_\rho$.}
\label{tab:sim_summary}
\begin{tabular}{lll}
\toprule

\textbf{Quantity} & \textbf{Formula
}& \textbf{Level} \\
\midrule

Between-cluster nuisance & $w_j\sim\mathcal{N}(0,1)$ & between \\

Latent variable & $\ell_{ij}=f_{0.60}(w_j)$ & within ($\ell_{ij}$) / between ($\ell_{j}$)\\

Within-cluster nuisance & $z_{ij}=f_{0.60}(\ell_{ij})$ & within\\

Target proxy & $x_{kij}=f_{0.85}(\ell_{ij})$ & same as the latent variable\\

\bottomrule
\end{tabular}
\end{table}


Inference on the effect of interest at significance level $\alpha=0.05$ is then carried out using PIMAX, with first-stage summary statistics estimated by Firth-corrected logistic regression \citep{firth1993bias}. This choice aims at reducing small-sample bias and separation issues in the cluster-specific binary regressions. At the multiverse level, the global null hypothesis of absence of any effect is tested with the test of Definition \ref{def:global_combined_statistic}, taking $B=1000$ transformations and both $\psi=\text{mean}$ and $\psi=\max$ as combining functions. As described above, the competing GLMM-based inferential procedure is applied separately to each specification and adjusted across specifications using Holm-Bonferroni. The fitted GLMM includes a random intercept and a random slope for $z_{ij}$; in the within-cluster scenario, it additionally includes a random slope for $x_{kij}$.

First, the type I error is evaluated under $\beta=0$, with results displayed in Figures \ref{fig:error_within} and \ref{fig:error_between}. In both settings, GLMM-based inference exhibits substantial inflation of the type I error rate, especially for smaller values of $J$, with the problem being particularly severe in the between-cluster case. PIMAX, in both its mean and max versions, never exceeds the nominal level, after accounting for the $95\%$ simulation confidence interval.

\begin{figure}
    \centering
    \includegraphics[width=.8\linewidth]{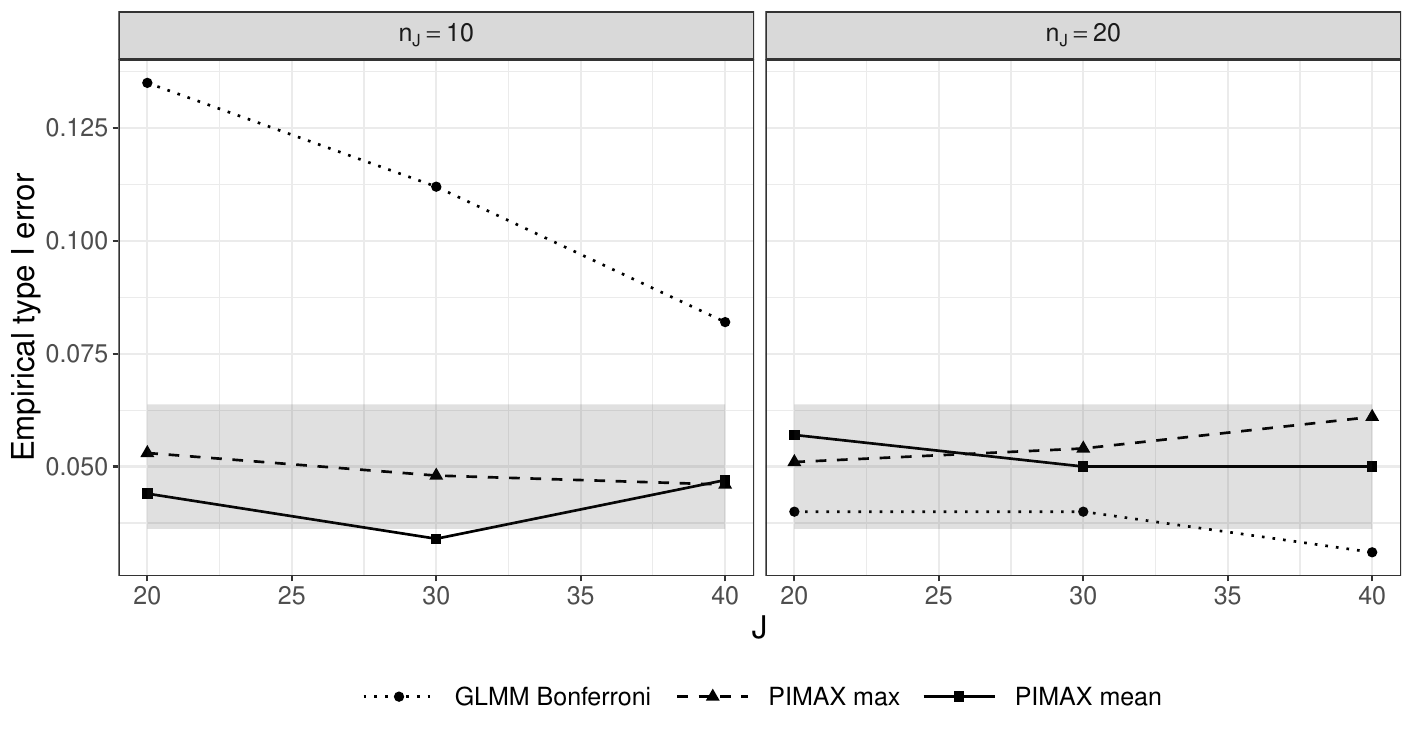}
    \caption{Simulations: empirical type I error in the within-cluster scenario. The shaded horizontal band represents the 95\% simulation interval around the nominal level $\alpha=0.05$.}
    \label{fig:error_within}
\end{figure}

\begin{figure}
    \centering
    \includegraphics[width=.8\linewidth]{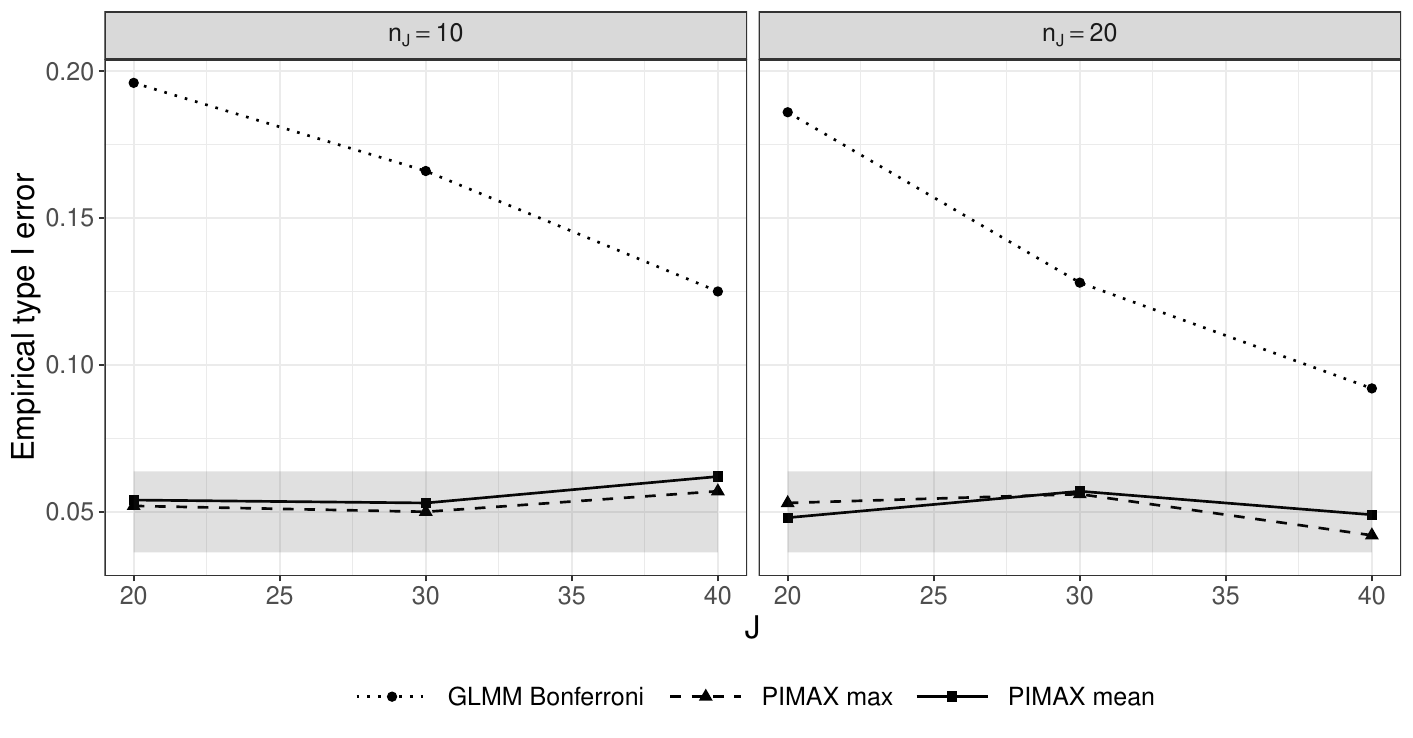}
    \caption{Simulations: empirical type I error in the between-cluster scenario. The shaded horizontal band represents the 95\% simulation interval around the nominal level $\alpha=0.05$.}
    \label{fig:error_between}
\end{figure}

Subsequently, statistical power is assessed by taking $\beta=0.5$ and reporting the proportion of simulations where the global null hypothesis is rejected, as shown in Figures \ref{fig:power_within} and \ref{fig:power_between}. The mean version of PIMAX is uniformly the most powerful procedure. The max version is more conservative, but still remains competitive in many settings; GLMM-based inference is both valid and superior to the latter only in the within-cluster case with $n_j=20$.

\begin{figure}
    \centering
    \includegraphics[width=.8\linewidth]{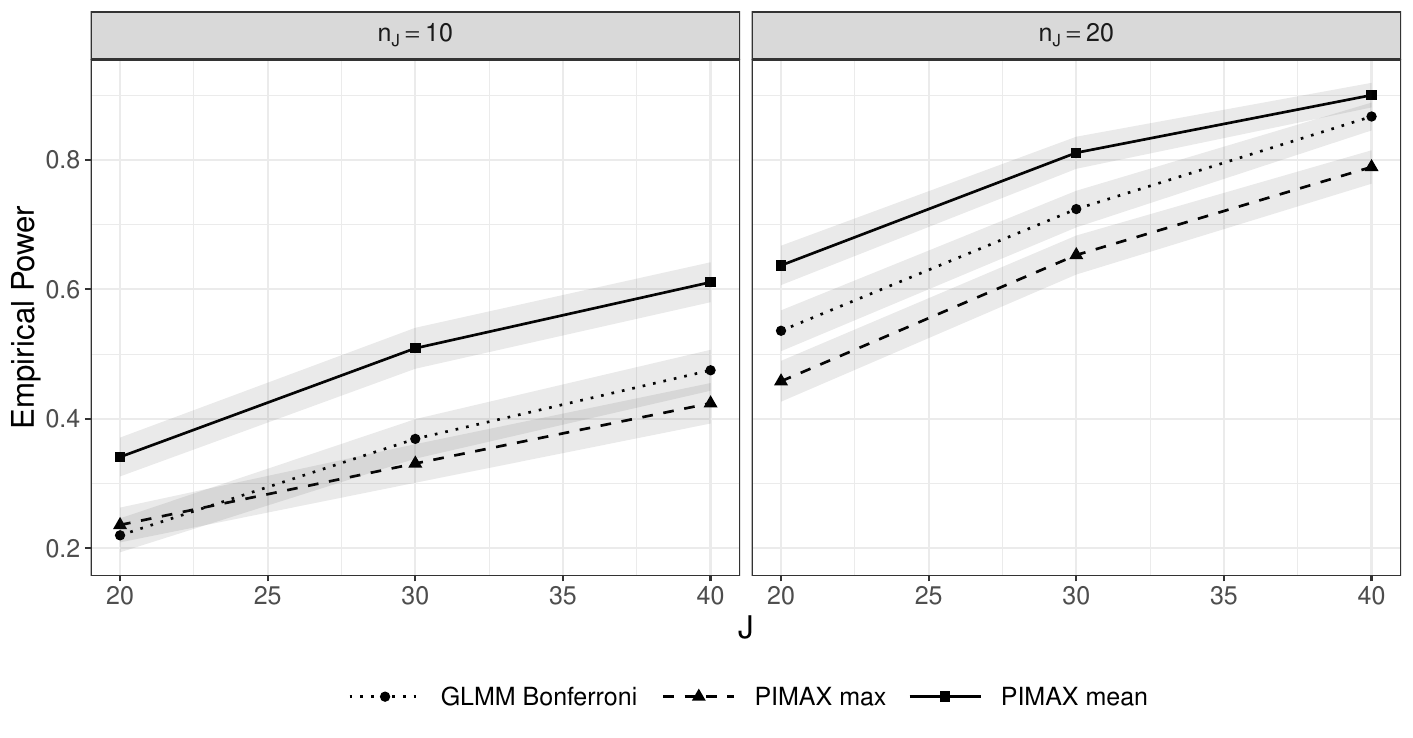}
    \caption{Simulations: empirical power in the within-cluster scenario. Shaded bands represent pointwise $95\%$ simulation intervals.}
    \label{fig:power_within}
\end{figure}

\begin{figure}
    \centering
    \includegraphics[width=.8\linewidth]{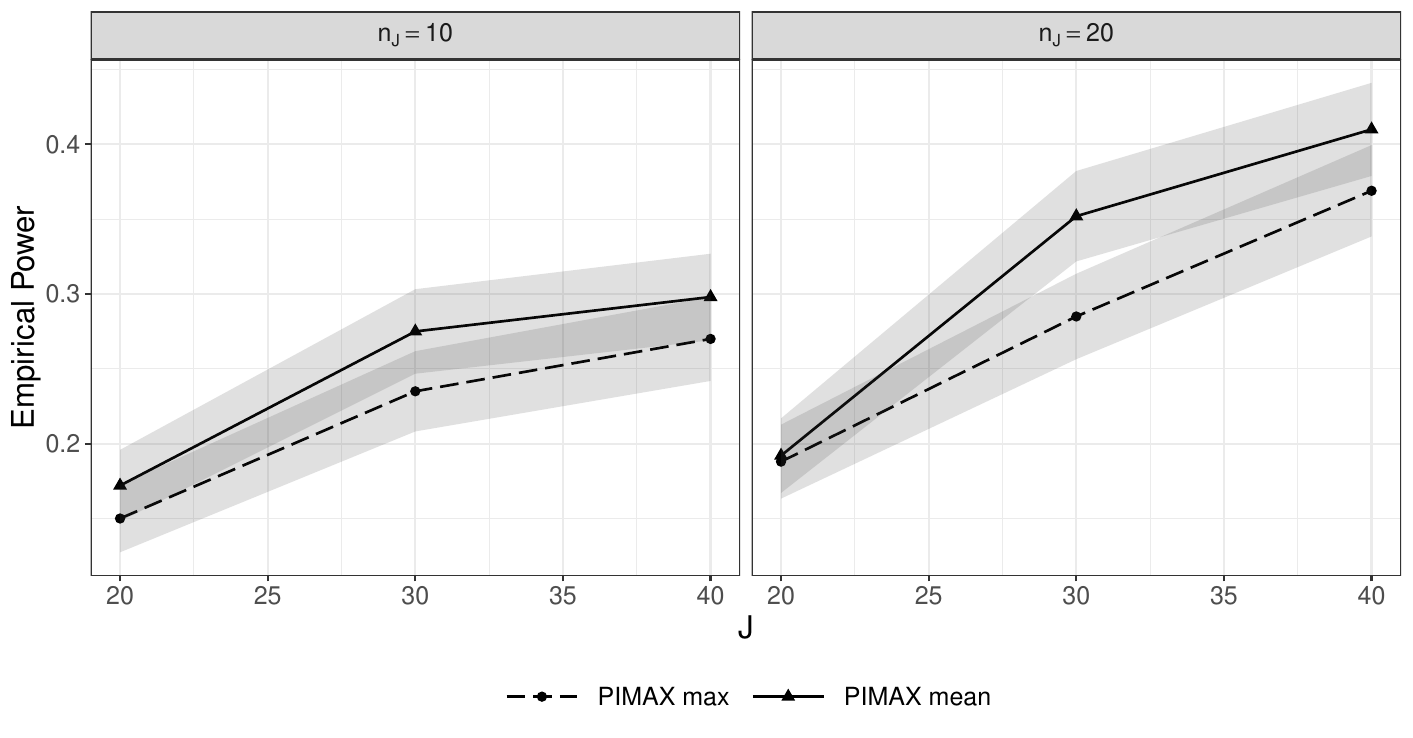}
    \caption{Simulations: empirical power in the between-cluster scenario. Shaded bands represent pointwise $95\%$ simulation intervals. Results for GLMM-based inference are omitted, as it fails to control the type I error rate in all considered settings.}
    \label{fig:power_between}
\end{figure}

Overall, the simulations suggest that PIMAX attains the intended balance between error control and sensitivity, with the mean combining function appearing particularly effective in the considered designs.

\section{SHARE data analysis}\label{application}

We illustrate the proposed methodology using data from SHARE \citep{borsch2013data}, a longitudinal survey collecting socio-demographic information on individuals aged 50 years and older in Europe. Our goal is to identify individual characteristics associated with disability among older adults in Italy. We study all available SHARE waves from 2004 to 2022 (waves 1--9), excluding wave 3, for which the imputed datasets required for the analysis are not available. The study is restricted to the main Italian respondents, yielding a sample of $5{,}325$ subjects. As is typical in longitudinal surveys, participants may enter the study after the first wave or leave before its conclusion, resulting in an unbalanced panel.

As the response, we take the global activity limitation indicator (GALI), a binary measure of long-term activity limitations due to health problems that is widely used as a proxy for disability in population health research \citep{van2018measuring}. As potential predictors, we examine sex, age, their interaction, education level, partnership status (capturing daily social support), financial distress (as a proxy for economic status), and chronic morbidity. These variables admit several operationalizations, as follows.

First, age is modeled either as a linear or a quadratic effect; more flexible parameterizations, such as spline-based effects or age categories, were not considered because they frequently resulted in unstable estimates or convergence issues. Second, education is reported according to the international standard classification of education \citep[ISCED]{unesco2003isced}, which distinguishes several ordered levels of educational attainment. We consider two codings: a three-level ordinal variable (primary, secondary, and post-secondary) and a binary indicator for having at least a secondary education. Partnership status is derived from the SHARE marital-status variable, which distinguishes several legal and living arrangements; we examine the presence of either any partner or a cohabiting partner. Financial distress is measured using the question on how easily households make ends meet, with four ordered response categories ranging from living comfortably to experiencing great difficulty. This variable is dichotomized to account for either any financial difficulty or only severe distress. Finally, chronic morbidity is specified in three different ways: the total count of chronic illnesses, an indicator for at least one condition, and an indicator for at least two conditions, capturing alternative cumulative and threshold effects.

With the exception of age, the predictors exhibit little or no within-subject variation over time and can be regarded as approximately time-invariant (i.e., between-subject) characteristics. To improve estimation stability and simplify model specification, we therefore summarize these covariates at the subject level. Specifically, for each individual, we use the average number of chronic illnesses across observed waves and the median value of the remaining covariates, corresponding to the individual's typical state during follow-up. Age is the only time-varying covariate and thus the only within-subject predictor.

The resulting multiverse consists of 48 models, each corresponding to a distinct combination of the alternative variable specifications described above. Since inference is performed for all regression coefficients, the analysis comprises a total of 408 hypothesis tests. We apply PIMAX using Firth-corrected logistic regression in the first stage, and $B=5000$ sign-flip transformations in the second stage, obtaining multiplicity-adjusted $p$-values for each model-coefficient pair via the maxT shortcut procedure. We set $\alpha=0.05$. The multiverse could be expanded further by including additional modeling choices, such as alternative summary statistics. Here, we focus on estimates from Firth-corrected logistic regression, since standard maximum likelihood estimation in logistic regression can be unstable in small-sample settings with sparse binary outcomes or quasi-complete separation \citep{firth1993bias,heinze2006comparative}.

The estimated effects vary substantially across specifications, highlighting the importance of accounting for specification uncertainty. Figure \ref{fig:share_chronic} reports the estimated effects of chronic morbidity under the different models, where, depending on the specification of chronic morbidity itself and the other covariates, the estimated association with the GALI ranges from the expected positive effect to an apparently protective (negative) effect. This variability illustrates how conventional analyses based on a single specification may lead to substantially different scientific conclusions. In some specifications, the estimated effects become particularly counterintuitive: for example, Table \ref{tab:share_model} reports a model in which a higher number of chronic illnesses is associated with lower disability, while higher educational attainment appears to increase the GALI score.

\begin{figure}
    \centering
    \includegraphics[width=.8\linewidth]{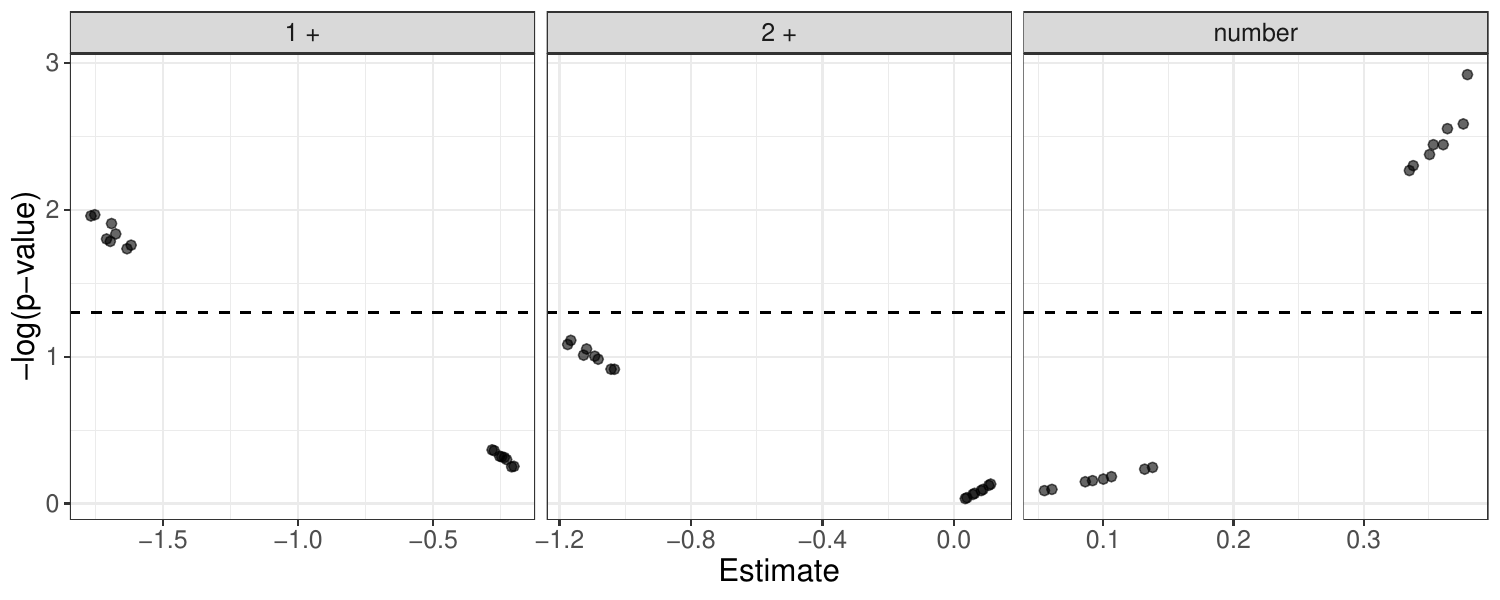}
    \caption{SHARE data: estimated effect of chronic morbidity across the multiverse versus $-\log_{10}(p)$, where $p$ denotes the raw $p$-value. The horizontal line indicates the significance level $\alpha=0.05$; points above it correspond to significant effects.}
    \label{fig:share_chronic}
\end{figure}

\begin{table}
\centering
\caption{SHARE data: coefficient estimates and raw $p$-values from a single selected model specification.}
\label{tab:share_model}
\begin{tabular}{lrr}
\toprule
\textbf{Predictor} & \textbf{Estimate} & \textbf{$p$-value}\\
\midrule
sex     & $2.34$ & $0.0002$ \\
age (linear)    & $0.07$ & $0.0002$ \\
sex:age     & $-0.03$ & $0.0004$ \\
education (secondary)     & $0.62$ & $0.3712$ \\
education (post-secondary)     & $2.22$ & $0.0428$ \\
partner (cohabiting)     & $-0.23$ & $0.7360$ \\
financial stress (medium/strong)   & $1.79$ & $0.0158$ \\
chronic ($1+$)   & $-1.71$ & $0.0064$ \\
\bottomrule
\end{tabular}
\end{table}

Table \ref{tab:share_summary} summarizes the results of the full analysis by reporting, for each predictor, the percentage of model specifications in which its effect is declared significant before and after multiplicity adjustment. After multiplicity correction, only age and sex remain significant in most model specifications. This is consistent with epidemiological evidence identifying older age and female sex as the most robust predictors of disability in later life \citep{istat2021anziani}. In contrast, the evidence for the remaining socio-economic and health-related predictors is substantially attenuated once specification uncertainty is accounted for, suggesting that their associations are less robust and more sensitive to modeling choices.

\begin{table}
\centering
\caption{SHARE data: proportion of model specifications in which each predictor is significant before (raw $p$-values) and after multiplicity adjustment (PIMAX).}
\label{tab:share_summary}
\begin{tabular}{lrr}
\toprule
\textbf{Predictor} & $\boldsymbol{\%}$ \textbf{Models (raw)} & $\boldsymbol{\%}$ \textbf{Models (adj.)}\\
\midrule
sex  & $100.0$ & $70.8$ \\
age     & $100.0$ & $100.0$ \\
sex : age & $100.0$ & $0.0$ \\
education     & $20.8$ & $0.0$ \\
partner & $0.0$ & $0.0$ \\
financial stress  & $50.0$ & $0.0$ \\
chronic & $33.3$ & $0.0$ \\
\bottomrule
\end{tabular}
\end{table}

We underline that this application is intended as a simple illustrative example to demonstrate how the proposed method operates in practice. However, as emphasized by \citet{girardi2024post}, multiverse analysis is not an invitation to indiscriminately expand the model space. A proper multiverse should not be constructed from all possible combinations of covariate definitions, but rather restricted to a small set of scientifically plausible models that are expected to reasonably approximate the data-generating reality. In other words, one must think carefully before testing.

\section{Discussion}\label{discussion}

We propose PIMAX, a method for post-selection inference in multiverse analysis with clustered observations. The method combines the two-stage summary-statistics framework of flip2sss \citep{andreella2025robust} with the
multiverse inference approach of PIMA \citep{girardi2024post}. In this way, PIMAX extends PIMA to settings where observations within the same cluster may be dependent. The method handles unbalanced clusters and heteroscedasticity, and does not require specifying a full random-effects covariance structure.

From a methodological perspective, the main contribution of PIMAX is to separate the two sources of complexity that typically occur together in clustered-data applications: dependence and model multiplicity. Dependence is handled through the cluster-level reduction induced by the first-stage summaries, whereas multiplicity is addressed through common sign-flip transformations across specifications and the closed-testing principle. This yields a framework that supports both global and post-selection inference across a multiverse of plausible mixed models.

A useful feature of our proposal is that it avoids making inference depend on a single, fully specified random-effects structure. This is important because the specification of the random component in mixed models is often difficult, can be unstable, and is rarely unique in practice. PIMAX shifts the focus to the fixed-effects specifications that define the scientific multiverse, while accounting for clustering through its two-stage construction. In this sense, it allows researchers to explore different plausible fixed-effects models without having to resolve the random-effects specification problem within each candidate model.

As with related resampling-based and post-selection methods, some points should be kept in mind. The theoretical results are asymptotic in the number of clusters, so the method is expected to provide a better approximation and more accurate inference when the number of clusters is sufficiently large. The performance also depends on the first-stage summaries. If these summaries are unstable or strongly biased, for example, in very small or sparse clusters, the second-stage test may be affected. Finally, as in any multiverse analysis, the set of candidate specifications and all components of the multiverse design (including first-stage summaries, number of sign-flip transformations, combining function, and significance level) should be defined a priori, before any results are inspected \citep{girardi2024post}. These requirements are not specific to PIMAX, but reflect standard conditions for valid resampling-based and post-selection inference. Finally, future work could investigate extensions of PIMAX to more complex forms of dependence, including crossed random-effects structures.




\section*{Declarations}

\paragraph{Funding}
Anna Vesely acknowledges financial support from the Italian Complementary National Plan PNC-I.1 “Research initiatives for innovative technologies and pathways in the health and welfare sector” D.D.~931 of 06/06/2022, “DARE - DigitAl lifelong pRevEntion" initiative, code PNC0000002, CUP: B53C22006450001.

\paragraph{Conflict of interest}
The authors declare no conflict of interest.

\paragraph{Ethics approval and consent to participate}
Not applicable.

\paragraph{Consent for publication}
Not applicable.

\paragraph{Data availability}
The SHARE data used in this study are publicly available from the SHARE Research Data Center: \url{https://share-eric.eu/}.

\paragraph{Code availability}
The code used for the simulation study and the SHARE data analysis is available at: \url{https://github.com/angeella/PIMAX}.

\paragraph{Author contributions}
Both authors contributed equally to this work.

\begin{appendices}

\section{Proofs}\label{secA1}

\subsection{Proposition \ref{prop:null_equivalence}}

\begin{proof}
We first prove that $H_0:\beta=0$ implies
$\widetilde H_0:\widetilde\beta=0$. Under $H_0$, the covariate of interest has
no systematic effect on the response. Hence, in the within-cluster case, where
$t_j=\widehat\beta_j$, the population summaries have no systematic component
associated with the effect of $x_{ij}$. In the between-cluster case, where
$t_j=\widehat\xi_j$ and $q_j=x_j$, the population summaries do not depend
systematically on $x_j$.

By Assumption \ref{ass2},
\[
    \mathbb E(t_j\mid q_j,\boldsymbol u_j)
    =
    \widetilde\beta q_j+\boldsymbol u_j^\top\boldsymbol\varphi .
\]
Therefore, if the population summaries contain no systematic target component,
the coefficient of $q_j$ in this conditional mean must be zero, that is
$\widetilde\beta=0$.

Conversely, suppose that $\widetilde\beta=0$. Then, by Assumption
\ref{ass2},
\[
    \mathbb E(t_j\mid q_j,\boldsymbol u_j)
    =
    \boldsymbol u_j^\top\boldsymbol\varphi ,
\]
so the conditional mean of the cluster summaries contains no systematic component associated with $q_j$. Since, by construction, $q_j$ represents the effect of the covariate of interest in the second-stage model, this implies absence of the corresponding effect in the original model, namely $\beta=0$.
\end{proof}

\subsection{Proposition \ref{prop:bias}}

\begin{proof}
Since $f_j$ is linear, we can write $f_j(t)=a_jt+b_j$, with
$C=\sup_j |a_j|<\infty$. By definition,
$S^\ast = J^{-1/2}\sum_{j=1}^J f_j(t_j^\ast)$. Therefore,
\begin{align*}
\mathbb{E}(S)-\mathbb{E}(S^\ast)
&=
J^{-1/2}\sum_{j=1}^J
\left\{
\mathbb{E}[f_j(t_j)]-f_j(t_j^\ast)
\right\}  
=
J^{-1/2}\sum_{j=1}^J
a_j\{\mathbb{E}(t_j)-t_j^\ast\},
\end{align*}
and so
\[
\left|\mathbb{E}(S)-\mathbb{E}(S^\ast)\right|
\leq
C J^{-1/2}\sum_{j=1}^J
\left|\mathbb{E}(t_j)-t_j^\ast\right|.
\]
The right-hand side converges to zero by assumption, and therefore
$\mathbb{E}(S)-\mathbb{E}(S^\ast)\to 0$, as $J\to\infty$.
\end{proof}

\subsection{Theorem \ref{thm:local_validity}}

\begin{proof}
By Proposition \ref{prop:null_equivalence}, under Assumption \ref{ass2},
$H_0:\beta=0$ is equivalent to
$\widetilde H_0:\widetilde\beta=0$ in the second-stage working model.
It is therefore sufficient to prove the validity of the second-stage score test. Under Assumption \ref{ass1}, the cluster-level score contributions are mutually independent. Assumption \ref{ass3} gives the Lindeberg condition and a non-degenerate limiting variance for the oracle effective score.
Thus, applying Theorem 2 from \citet{de2025inference} to the cluster-level second-stage score contributions, the sign-flipping distribution consistently approximates the null distribution of the standardized score statistic under $\widetilde H_0$. Hence the rejection rule in Definition \ref{def:test} has asymptotic level $\alpha$ under
$\widetilde H_0$, and therefore, by Proposition \ref{prop:null_equivalence}, under $H_0$.
\end{proof}

\subsection{Theorem \ref{thm:pimax_validity}}

\begin{proof}
Under $H_{\mathcal M}=\bigcap_{k=1}^K H_{0k}$, all model-specific null hypotheses are true. By Theorem \ref{thm:local_validity}, the
standardized sign-flipping score statistic is asymptotically valid for each specification. Definition \ref{def:global_combined_statistic} applies the same cluster-level sign-flips to all specifications and combines the resulting statistics through a pre-specified non-decreasing function $\psi$. Therefore, Theorem 3.2 from \citet{de2025permutation} applies, with clusters as independent units and specifications as parallel models. It follows that the sign-flipping distribution of $G_b$ consistently approximates the null distribution of $G_1$ under $H_{\mathcal M}$, and the test has asymptotic level $\alpha$.
\end{proof}





\end{appendices}

\bibliographystyle{apalike} 
\bibliography{sn-bibliography.bib}

@article{andreella2025robust,
  title={Robust inference for generalized linear mixed models: a “two-stage summary statistics” approach based on score sign flipping},
  author={Andreella, Angela and Goeman, Jelle and Hemerik, Jesse and Finos, Livio},
  journal={Psychometrika},
  volume={90},
  number={2},
  pages={531--553},
  year={2025},
  publisher={Cambridge University Press}
}

@article{goeman2011multiple,
  title={Multiple testing for exploratory research},
  author={Goeman, Jelle J and Solari, Aldo},
  journal={Statistical Science},
  pages={584--597},
  year={2011},
  publisher={JSTOR}
}

@book{westfall1993resampling,
  author    = {Westfall, Peter H. and Young, S. Stanley},
  title     = {Resampling-Based Multiple Testing: Examples and Methods for p-Value Adjustment},
  publisher = {John Wiley \& Sons},
  address   = {New York},
  year      = {1993},
  isbn      = {978-0-471-55761-6}
}

@article{vesely2023permutation,
  title={Permutation-based true discovery guarantee by sum tests},
  author={Vesely, Anna and Finos, Livio and Goeman, Jelle J},
  journal={Journal of the Royal Statistical Society Series B: Statistical Methodology},
  volume={85},
  number={3},
  pages={664--683},
  year={2023},
  publisher={Oxford University Press US}
}

@article{bolker2009generalized,
  title={Generalized linear mixed models: a practical guide for ecology and evolution},
  author={Bolker, Benjamin M and Brooks, Mollie E and Clark, Connie J and Geange, Shane W and Poulsen, John R and Stevens, M Henry H and White, Jada-Simone S},
  journal={Trends in ecology \& evolution},
  volume={24},
  number={3},
  pages={127--135},
  year={2009},
  publisher={Elsevier}
}

@article{goeman2021only,
  title={Only closed testing procedures are admissible for controlling false discovery proportions},
  author={Goeman, Jelle J and Hemerik, Jesse and Solari, Aldo},
  journal={The Annals of Statistics},
  volume={49},
  number={2},
  pages={1218--1238},
  year={2021},
  publisher={JSTOR}
}

@article{lin2010relative,
  title={On the relative efficiency of using summary statistics versus individual-level data in meta-analysis},
  author={Lin, Dan-Yu and Zeng, Daniel},
  journal={Biometrika},
  volume={97},
  number={2},
  pages={321--332},
  year={2010},
  publisher={Oxford University Press}
}

@article{senn1990analysis,
  title={Analysis of serial measurements in medical research},
  author={Senn, Stephen},
  journal={BMJ: British Medical Journal},
  volume={300},
  number={6725},
  pages={680},
  year={1990}
}

@inproceedings{andreella2025blockwise,
  title={Blockwise Resampling for Robust Fixed Effects Inference in Linear Mixed Models},
  author={Andreella, Angela and Finos, Livio},
  booktitle={Scientific Meeting of the Italian Statistical Society},
  pages={46--51},
  year={2025},
  organization={Springer}
}

@article{heagerty2001misspecified,
  title={Misspecified maximum likelihood estimates and generalised linear mixed models},
  author={Heagerty, Patrick J and Kurland, Brenda F},
  journal={Biometrika},
  volume={88},
  number={4},
  pages={973--985},
  year={2001},
  publisher={Oxford University Press}
}

@incollection{unesco2003isced,
  author    = {{UNESCO}},
  title     = {{International Standard Classification of Education, ISCED 1997}},
  booktitle = {Advances in Cross-National Comparison: A European Working Book for Demographic and Socio-Economic Variables},
  editor    = {Hoffmeyer-Zlotnik, J{\"u}rgen H. P. and Wolf, Christof},
  pages     = {195--220},
  publisher = {Springer},
  address   = {Boston, MA},
  year      = {2003}
}

@techreport{istat2021anziani,
  author      = {{Istat}},
  title       = {Le condizioni di salute della popolazione anziana in {I}talia},
  institution = {Istituto Nazionale di Statistica},
  year        = {2021},
  address     = {Roma}
}

@article{van2018measuring,
  title={Measuring disability: a systematic review of the validity and reliability of the Global Activity Limitations Indicator ({GALI})},
  author={Van Oyen, Herman and Bogaert, Petronille and Yokota, Renata TC and Berger, Nicolas},
  journal={Archives of Public Health},
  volume={76},
  number={1},
  pages={25},
  year={2018},
  publisher={Springer}
}

@article{de2025permutation,
  title={Permutation-based multiple testing when fitting many generalized linear models},
  author={De Santis, Riccardo and Goeman, Jelle J and Davenport, Samuel and Hemerik, Jesse and Finos, Livio},
  journal={Electronic Journal of Statistics},
  volume={19},
  number={2},
  pages={3317--3332},
  year={2025},
  publisher={The Institute of Mathematical Statistics and the Bernoulli Society}
}

@article{de2025inference,
  title={Inference in generalized linear models with robustness to misspecified variances},
  author={De Santis, Riccardo and Goeman, Jelle J and Hemerik, Jesse and Davenport, Samuel and Finos, Livio},
  journal={Journal of the American Statistical Association},
  volume={120},
  number={552},
  pages={2762--2771},
  year={2025},
  publisher={Taylor \& Francis}
}

@article{marcus1976closed,
  title={On closed testing procedures with special reference to ordered analysis of variance},
  author={Marcus, Ruth and Eric, Peritz and Gabriel, K Ruben},
  journal={Biometrika},
  volume={63},
  number={3},
  pages={655--660},
  year={1976},
  publisher={Oxford University Press}
}

@book{pesarin2001multivariate,
  author    = {Pesarin, Fortunato},
  title     = {Multivariate Permutation Tests: With Applications in Biostatistics},
  publisher = {John Wiley \& Sons},
  address   = {Chichester},
  year      = {2001},
  isbn      = {978-0-471-49670-0}
}

@article{borsch2013data,
  title={Data resource profile: the {S}urvey of {H}ealth, {A}geing and {R}etirement in {E}urope ({SHARE})},
  author={B{\"o}rsch-Supan, Axel and Brandt, Martina and Hunkler, Christian and Kneip, Thorsten and Korbmacher, Julie and Malter, Frederic and Schaan, Barbara and Stuck, Stephanie and Zuber, Sabrina},
  journal={International journal of epidemiology},
  volume={42},
  number={4},
  pages={992--1001},
  year={2013},
  publisher={Oxford University Press}
}

@article{breslow1993approximate,
  author  = {Breslow, Norman E. and Clayton, David G.},
  title   = {Approximate Inference in Generalized Linear Mixed Models},
  journal = {Journal of the American Statistical Association},
  year    = {1993},
  volume  = {88},
  number  = {421},
  pages   = {9--25}
}

@book{mcculloch2008generalized,
  author    = {McCulloch, Charles E. and Searle, Shayle R. and Neuhaus, John M.},
  title     = {Generalized, Linear, and Mixed Models},
  edition   = {2},
  series    = {Wiley Series in Probability and Statistics},
  publisher = {John Wiley \& Sons},
  address   = {Hoboken, NJ},
  year      = {2008},
  isbn      = {978-0-470-07371-1}
}

@article{gelman2014garden,
  author  = {Gelman, Andrew and Loken, Eric},
  title   = {The Statistical Crisis in Science},
  journal = {American Scientist},
  year    = {2014},
  volume  = {102},
  number  = {6},
  pages   = {460--465}
}

@article{steegen2016increasing,
  author  = {Steegen, Sara and Tuerlinckx, Francis and Gelman, Andrew and Vanpaemel, Wolf},
  title   = {Increasing Transparency Through a Multiverse Analysis},
  journal = {Perspectives on Psychological Science},
  year    = {2016},
  volume  = {11},
  number  = {5},
  pages   = {702--712}
}

@article{simmons2011false,
  author  = {Simmons, Joseph P. and Nelson, Leif D. and Simonsohn, Uri},
  title   = {False-Positive Psychology: Undisclosed Flexibility in Data Collection and Analysis Allows Presenting Anything as Significant},
  journal = {Psychological Science},
  year    = {2011},
  volume  = {22},
  number  = {11},
  pages   = {1359--1366}
}

@article{nosek2014scientific,
  author  = {Nosek, Brian A. and Lakens, Dani\"el},
  title   = {Registered Reports: A Method to Increase the Credibility of Published Results},
  journal = {Social Psychology},
  year    = {2014},
  volume  = {45},
  number  = {3},
  pages   = {137--141}
}

@article{girardi2024post,
  author  = {Girardi, Paolo and Vesely, Anna and Lakens, Dani\"el and Alto\`e, Gianmarco and Pastore, Massimiliano and Calcagn\`i, Antonio and Finos, Livio},
  title   = {Post-selection {I}nference in {M}ultiverse {A}nalysis ({PIMA}): An Inferential Framework Based on the Sign Flipping Score Test},
  journal = {Psychometrika},
  year    = {2024},
  volume  = {89},
  number  = {2},
  pages   = {542--568}
}

@article{hemerik2020robust,
  author  = {Hemerik, Jesse and Goeman, Jelle J. and Finos, Livio},
  title   = {Robust Testing in Generalized Linear Models by Sign Flipping Score Contributions},
  journal = {Journal of the Royal Statistical Society: Series B},
  year    = {2020},
  volume  = {82},
  number  = {3},
  pages   = {841--864}
}

@article{cnaan1997tutorial,
  author  = {Cnaan, Avital and Laird, Nan and Slasor, Peter},
  title   = {Tutorial in Biostatistics: Using the General Linear Mixed Model to Analyse Unbalanced Repeated Measures and Longitudinal Data},
  journal = {Statistics in Medicine},
  year    = {1997},
  volume  = {16},
  number  = {20},
  pages   = {2349--2380}
}

@Manual{pima,
    title = {pima: Post-selection Inference in Multiverse Analysis},
    author = {Livio Finos and Paolo Girardi and Filippo Gambarota and Anna Vesely and Giulia Calignano and Massimiliano Pastore and Gianmarco Altoè},
    year = {2022},
    url = {https://github.com/livioivil/pima},
    note = {R package}
  }

@article{holm1979simple,
  title={A simple sequentially rejective multiple test procedure},
  author={Holm, Sture},
  journal={Scandinavian Journal of Statistics},
  pages={65--70},
  year={1979},
  publisher={JSTOR}
}

@article{benjamini2020selective,
  title={Selective inference: The silent killer of replicability},
  author={Benjamini, Yoav},
  journal={Harvard Data Science Review},
  volume={2},
  number={4},
  year={2020},
  publisher={The MIT Press}
}

@article{andreella2026multivariate,
  title={Multivariate mixed models with model-free random effects},
  author={Andreella, Angela and Finos, Livio},
  journal={arXiv preprint arXiv:2604.27907},
  year={2026}
}

@article{blanchard2020post,
  title   = {Post hoc confidence bounds on false positives using reference families},
  author  = {Blanchard, Gilles and Neuvial, Pierre and Roquain, Etienne},
  journal = {The Annals of Statistics},
  year    = {2020},
  volume  = {48},
  number  = {3},
  pages   = {1281--1303}
}

@article{blain2022notip,
  title={Notip: Non-parametric true discovery proportion control for brain imaging},
  author={Blain, Alexandre and Thirion, Bertrand and Neuvial, Pierre},
  journal={NeuroImage},
  volume={260},
  pages={119492},
  year={2022},
  publisher={Elsevier}
}

@article{andreella2023permutation,
  title={Permutation-based true discovery proportions for functional magnetic resonance imaging cluster analysis},
  author={Andreella, Angela and Hemerik, Jesse and Finos, Livio and Weeda, Wouter and Goeman, Jelle},
  journal={Statistics in Medicine},
  volume={42},
  number={14},
  pages={2311--2340},
  year={2023},
  publisher={Wiley Online Library}
}

@article{firth1993bias,
  title={Bias reduction of maximum likelihood estimates},
  author={Firth, David},
  journal={Biometrika},
  pages={27--38},
  year={1993},
  publisher={JSTOR}
}

@article{barr2013random,
  author  = {Barr, Dale J. and Levy, Roger and Scheepers, Christoph and Tily, Harry J.},
  title   = {Random Effects Structure for Confirmatory Hypothesis Testing: Keep It Maximal},
  journal = {Journal of Memory and Language},
  year    = {2013},
  volume  = {68},
  number  = {3},
  pages   = {255--278}
}

@article{bates2015parsimonious,
  author  = {Bates, Douglas and Kliegl, Reinhold and Vasishth, Shravan and Baayen, Harald},
  title   = {Parsimonious Mixed Models},
  journal = {arXiv preprint arXiv:1506.04967},
  year    = {2015}
}

@article{matuschek2017balancing,
  author  = {Matuschek, Hannes and Kliegl, Reinhold and Vasishth, Shravan and Baayen, Harald and Bates, Douglas},
  title   = {Balancing Type {I} Error and Power in Linear Mixed Models},
  journal = {Journal of Memory and Language},
  year    = {2017},
  volume  = {94},
  pages   = {305--315}
}

@article{simonsohn2020specification,
  title={Specification curve analysis},
  author={Simonsohn, Uri and Simmons, Joseph P and Nelson, Leif D},
  journal={Nature Human Behaviour},
  volume={4},
  number={11},
  pages={1208--1214},
  year={2020},
  publisher={Nature Publishing Group}
}

@article{heinze2006comparative,
author = {Heinze, Georg},
title = {A comparative investigation of methods for logistic regression with separated or nearly separated data},
journal = {Statistics in Medicine},
volume = {25},
number = {24},
pages = {4216-4226},
year = {2006}
}

\end{document}